\documentclass{article}
\usepackage{amsmath, amssymb, amsthm}
  \usepackage{paralist}
  \usepackage{graphics} 
  \usepackage{epsfig} 
 \usepackage[colorlinks=true]{hyperref}

\hypersetup{urlcolor=blue, citecolor=red}

\newtheorem{thm}{Theorem}[section]

\newtheorem{lem}[thm]{Lemma}
\newtheorem{prop}{Proposition}

\newtheorem{defi}[thm]{Definition}
\newtheorem{rem}{Remark}

\newcommand{\ep}{\varepsilon}

\newcommand{\R}{\mathbb{R}}

\newcommand{\ud}{\mathrm{d}}
\newcommand{\mc}{\mathcal}
\newcommand{\be}{\begin{equation}}
\newcommand{\ee}{\end{equation}}
\newcommand{\bs}{\begin{split}}
\newcommand{\es}{\end{split}}
\newcommand{\bee}{\begin{equation*}}
\newcommand{\eee}{\end{equation*}}

\newcommand{\ssi}{\Leftrightarrow}
\newcommand{\vphi}{\varphi}

\begin{document}
\title{Derivation of a two-fluids model for a Bose gas from a quantum kinetic system\footnote{Primary: 76Y05; Secondary: 82C40, 82D50.}}
\author{Thibaut Allemand\footnote{DMA, \'Ecole Normale Sup\'erieure, 45 rue d'Ulm, 75230 Paris Cedex 05, France}}
\maketitle

\begin{abstract}
We formally derive the hydrodynamic limit of a system modelling  a bosons gas having a condensed part, made of a quantum kinetic and a Gross-Pitaevskii equation. The limit model, which is a two-fluids Euler system, is approximated by an isentropic system, which is then studied. We find in particular some conditions for the hyperbolicity, and we study the weak solutions. A numerical example is given at the end.
\end{abstract}

\section{Introduction}

There has been a lot of research done about Bose-Einstein condensates in mathematics and physics, especially since the Nobel Prize 2001 was awarded to Carl, Cornell and Ketterle who succeeded in creating one for the first time. In particular, mathematicians have been working on various models to see if the predictions of the physicists could be justified mathematically.

There are actually several types of models to describe a Bose-Einstein condensate \cite{bec}. The only one which is valid without any approximation is the atomistic one, that is, the linear Schr\"odinger equation for the P bodies density:
\be
\label{Nbodies}
 ih\partial_t \Psi_P=H_P\Psi_P
\ee
where $H_P$ is the Hamiltonian, i.e. the energy operator, of the system, and $h$ is Planck's constant. Other models can be found depending on the regime under consideration. At zero temperature, the gas is entirely condensated. Assuming  that it is in a non-dissipative trap, the one-body density is governed by the Gross-Pitaevskii equation:
\[
 ih\partial_t \Psi = -h^2 \Delta_x \Psi + U\Psi +|\Psi|^2\Psi
\]
where $U$ is the trapping potential. When the temperature of the gas is close to zero, the gas is composed of a condensate and a normal component. The condensate is still expected to be governed by a Gross-Pitaevskii-type equation, with coupling terms taking into account the mass and momentum exchanges with the normal component. The latter can be described in a probabilistic way using kinetic equations; in this case, a quantum Boltzmann equation, first proposed by Nordheim \cite{nordheim}, with additional terms. We could also use a fluid description for the normal part together with the usual Gross-Pitaevskii  equation for the condensate; still another possibility is to use a fluid description for both parts of the gas, leading to a two fluids model as predicted by Landau \cite{landau}. The way these models are related is represented in Figure \ref{models}.

\begin{figure}
\begin{center}
  \includegraphics[width=10cm]{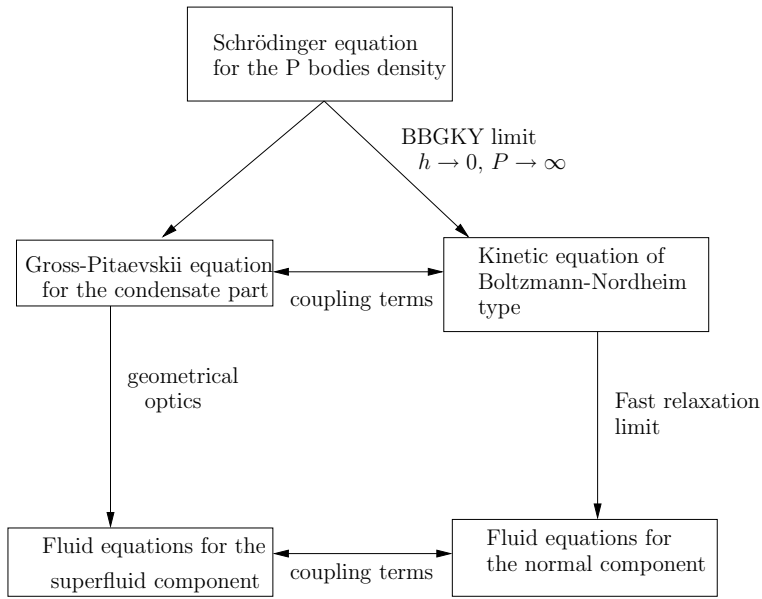}
\end{center}
\caption{Various models describing a Bose gas, and their interactions.}
\label{models}
\end{figure}

In \cite{pomeau}, a kinetic model for the non-condensate part has been derived from (\ref{Nbodies}) using the BBGKY expansion in the low density limit. The condensated part is then described by a Gross-Pitaevskii equation, which takes into account both the interactions between the atoms of the condensate, and between them and those of the non-condensate part (called normal fluid).

What we are interested in is to link this model with models at a larger scale, namely, fluid models. This can be done considering the fast relaxation limit for the kinetic equation, together with the semiclassical limit for the Gross-Pitaevskii equation.

Our starting point is the following model: let 
\[
 \varphi \equiv \varphi(t,x,v) \geq 0
\]
 denote the phase-space density function of the normal fluid, and 
\be
\label{psi}
 \Psi(t,x)=\sqrt{\rho_s(t,x)}e^{i\theta(t,x)/h}
\ee
denote the wave function of the condensated (or superfluid) phase, with $\rho_s$ being the mass of the superfluid. Here, $t \in \R_+$ is the time variable, and $x,v \in \R^N$ for $N \geq 1$ are the position and velocity variables. 

In \cite{pomeau}, Pomeau, Brachet, M\'etens and Rica derive the following coupled system of equations:

\begin{equation}
 \label{modele}
\begin{cases}
 & \displaystyle{\partial_t \varphi + v.\nabla_x \varphi - \nabla_x V_1.\nabla_v \varphi = Q(\varphi) + \rho_s Q_1(\varphi)} \\
 & \displaystyle{ih\partial_t \Psi = -\frac{h^2}{2}\Delta_x \Psi  + V_2\Psi + \frac{ih}{2}Q_2(\varphi)\Psi }.
\end{cases}
\end{equation}
The first equation expresses a balance between the left-hand term, which is a Vlasov operator for a gas in the external field $\nabla_x V_1$, and the right-hand term, which comes from the normal fluid-normal fluid and normal fluid-superfluid interactions. The second equation is a cubic Schr\"odinger equation (because of the form (\ref{potentials}) of $V_2$), with an interaction term with the normal fluid.  The first collision operator $Q(\vphi)$ is the usual Boltzmann-Nordheim operator for the bosons, modelling the collisions between two particles from the normal fluid. It writes
\begin{equation}
 \label{Q}
Q(\varphi) = \int_{\R^N}\int_{S^{N-1}} \left(\varphi'\varphi_*'(1+\varphi)(1+\varphi') - \varphi\varphi_*(1+\varphi')(1+\varphi_*')\right)\ud v_*\ud n
\end{equation}
with the usual notations
\bee
\varphi=\varphi(t,x,v),~~\varphi_*=\varphi(t,x,v_*),~~ \varphi'=\varphi(t,x,v'),~~ \varphi_*'=\varphi(t,x,v_*').
\eee
The post-collisional velocities $(v,v_*)$ and the precollisional ones $(v', v_*')$ satisfy the conservation of momentum and kinetic energy
\begin{equation}
\label{microscons}
 \begin{cases}
  & v+v_*=v'+v_*' \\
  & |v|^2+|v_*|^2=|v'|^2+|v_*'|^2
 \end{cases}
\end{equation}
since the collisions are assumed to be elastic. The solutions of (\ref{microscons}) can be parametrized by a unit vector $n\in S^{N-1}$, so that 
\bee
\begin{cases}
&v'=v-[(v-v_*).n]n\\
&v_*'=v_*+[(v-v_*).n]n.
\end{cases}
\eee
The other two collision operators $Q_1(\vphi)$ and $Q_2(\vphi)$ both take into account the collisions between particles from the normal fluid and from the superfluid. The operator $Q_1(\vphi)$ writes
\begin{equation}
\label{Q1}
\begin{split}
 Q_1(\varphi)(v) = \int_{C_1} \big[\vphi(v')&\vphi(v_*')-\vphi(v)\left(\vphi(v')+\vphi(v_*')+1\right)\big]\ud v'\ud v_*'\\
& +2 \int_{C_2} \big[\vphi(v_*')\left(\vphi(v)+\vphi(v_*)+1\right)-\vphi(v)\vphi(v_*)\big]\ud v_* \ud v_*'
\end{split}
\end{equation}
where $C_1$ is the set of all the velocities $v', v_*'$ such that
\begin{equation}
\label{cons1}
 \begin{cases}
  & \nabla_x \theta + v = v' + v_*' \\
  & -2 \partial_t \theta - 2V_1 + |v|^2 = |v'|^2+|v_*'|^2,
 \end{cases}
\end{equation}
and $C_2$ is the set of the velocities $v_*, v_*'$ satisfying
\begin{equation*}
 \begin{cases}
  & \nabla_x \theta + v_*' = v + v_* \\
  & -2 \partial_t \theta - 2V_1 + |v_*'|^2 = |v|^2+|v_*|^2.
 \end{cases}
\end{equation*}
As to the last collision operator $Q_2$, it writes
\begin{equation*}
 Q_2(\vphi) = \int_{C_3} \big[ \vphi(v')\vphi(v_*')-\vphi(v_*)(\vphi(v')+\vphi(v_*')+1)\big]\ud v_* \ud v' \ud v_*'
\end{equation*}
where $C_3$ is the set of all velocities $v_*, v', v_*'$ such that
\begin{equation*}
 \begin{cases}
  & \nabla_x \theta + v_* = v' + v_*' \\
  & -2 \partial_t \theta - 2V_1+ |v_*|^2 = |v'|^2+|v_*'|^2.
 \end{cases}
\end{equation*}
The terms $-2 \partial_t \theta - 2V_1$ which are added in the microscopic energy conservation equations are here to ensure the possibility of having a condensate with velocity and energy which are other than zero. Indeed, the usual Boltzmann-Nordheim equation was derived in a space homogeneous setting, and up to translations the energy and the velocity of the  condensate can only be zero in such a framework.

The potentials $V_1$ and $V_2$ are given by
\be
\label{potentials}
\begin{cases}
 &V_1 = \alpha(2\rho_n+\rho_s)\\
 &V_2 = \frac{\alpha}{2}(2\rho_n+\rho_s)
\end{cases}
\ee
where $\alpha$ is a positive constant, $\rho_s$ is given by (\ref{psi}) and
\[
 \rho_n = \int_{\R^N}\vphi(t,x,v)\ud v
\]
is the mass of the normal fluid. 

Note that very few is known mathematically about the Boltzmann equation for the bosons. It has been studied only in the spatially homogeneous case; see for example \cite{Lu}, \cite{EM01}, \cite{EMV03}.

In this work, we will (formally) perform the compressible Euler limit on the system (\ref{modele}) together with a semiclassical limit. To do this, we rescale time and space by changing $(t,x)$ to $(\frac{t}{\ep},\frac{x}{\ep})$, where $\ep$ is a small parameter, which gives the scaled system
\begin{equation}
 \label{modeleps}
\begin{cases}
 & \displaystyle{\partial_t \varphi + v.\nabla_x \varphi - \nabla_x V_1.\nabla_v \varphi = \frac{1}{\ep}Q(\varphi) + \frac{\rho_s}{\ep} Q_1(\varphi)} \\
 & \displaystyle{i\tilde h\partial_t \Psi = -\frac{\tilde h^2}{2}\Delta_x \Psi  + V_2\Psi + \frac{i\tilde h}{2\ep}Q_2(\varphi)\Psi }.
\end{cases}
\end{equation}
where $\tilde h=\ep h$, and we want to determine the asymptotics of this system as $\tilde h\to 0$ and $\ep\to 0$.

In Section \ref{properties}, we detail the main physical properties of the model. Then, in Section \ref{hydrodynamic limit}, we obtain a fluid description of the Bose-Einstein condensate and the non-condensate phase. We therefore study the fluid system. In Section \ref{hyperbolicity}, we get some conditions for the system to be hyperbolic, and then we show the existence of some local strong solutions. In Section \ref{weak solutions}, we finally give some details about weak solutions and the propagation of shocks.

\section{Properties of the model}
\label{properties}

The possibility of obtaining a fluid description of the Bose gas governed by (\ref{modele}) is a consequence of some fundamental properties of this system. First, like for the classical Boltzmann and Schr\"odinger equations, we can write the macroscopic equations wich govern the mass, momentum and energy for both parts of the gas. Then, we characterize the distributions $\vphi$ which minimize the entropy and thus are expected to be attractors of the dynamics in the fast relaxation limit. These two informations will allow us to perform the compressible Euler limit in Section \ref{hydrodynamic limit}.

\subsection{Conservation laws}

Integrating the kinetic equation in (\ref{modele}) against $1$, $v$, $|v|^2$, we obtain the equations on the mass, momentum and kinetic energy of the normal fluid. From the Gross-Pitaevskii equation, we derive the equations for the mass, momentum and energy of the condensate.

\subsubsection{Mass}
We plug the expression of $\Psi$ (\ref{psi}) in the second equation of the system (\ref{modele}). After dividing by the exponential, we can write two equations, one for the real part and one for the imaginary part. They read
\be
\label{massus}
  \partial_t \rho_s + \nabla_x.(\rho_s u_s)= \rho_s Q_2(\vphi)
\ee
and
\be
\label{phase}
  \partial_t \theta = \frac{h^2}{2\sqrt{\rho_s}}\Delta_x\sqrt{\rho_s}-\frac{|u_s|^2}{2}-V_2
\ee
where 
\[
 u_s = \nabla_x \theta
\]
denotes the bulk velocity of the condensed phase. Equation (\ref{massus}) gives the law governing the mass density of the condensate. 

For the kinetic equation, let us first remark that
\begin{equation}
\label{opcol}
 \int_{\R^N} \begin{pmatrix}
       1\\v\\|v|^2
      \end{pmatrix}
Q_1(\vphi)(v) \ud v = - \begin{pmatrix}
                      1\\ u_s \\ -2 \partial_t \theta - 2V_1
                     \end{pmatrix}
Q_2(\vphi).
\end{equation}
Then, integrating with respect to the velocity over the whole space $\R^N$, we obtain the equation satisfied by the mass of the normal fluid:
\be
\label{massun}
 \partial_t \rho_n + \nabla_x. (\rho_n u_n)=-\rho_s Q_2(\vphi)
\ee
where we defined the bulk velocity of the normal fluid as
\[
 u_n=\frac{1}{\rho_n}\int_{\R^N}vf(t,x,v)\ud v.
\]
We can see that neither $\rho_n$ nor $\rho_s$ are conserved quantities, unless the thermodynamic equilibrium is reached, that is,
\[
 Q_2(\vphi)=0.
\]
However, the total mass
\[
 \rho=\rho_n+\rho_s
\]
satisfies a conservation law without any assumption of equilibrium:
\be
\label{masstotal}
 \partial_t \rho + \nabla_x.(\rho_nu_n+\rho_su_s)=0.
\ee
This means that the normal fluid and the condensate exchange mass until they are at thermodynamic equilibrium.

\subsubsection{Momentum}

Integrating the kinetic equation against $v$, we get
\be
\label{momentumun}
 \partial_t (\rho_n u_n) + \nabla_x . \left(\int_{\R^N} v\otimes v \vphi \ud v\right) = -\rho_n \nabla_x V_1-\rho_s u_s Q_2(\vphi).
\ee
For the condensed part, we differentiate equation (\ref{phase}) with respect to $x$, and, thanks to the mass equations (\ref{massus}) and (\ref{massun}), we can write
\[
\begin{split}
 \partial_t(\rho_s u_s) + \nabla_x.(\rho_s u_s \otimes u_s)&= -\rho_s \nabla_x V_2+ \rho_s u_s Q_2(\vphi) \\
&~~~~+ \frac{h^2}{2}\nabla_x.\left( \sqrt{\rho_s}\nabla_x^2 \sqrt{\rho_s}-\nabla_x\sqrt{\rho_s}\otimes \nabla_x\sqrt{\rho_s}\right).
\end{split}
\]
Plugging the expression of the potentials (\ref{potentials}), we can see that the total momentum satisfies
\bee
\begin{split}
 \partial_t (\rho_nu_n+\rho_su_s)& +\nabla_x.\left(\int_{\R^N} v\otimes v \vphi \ud v +\rho_s u_s \otimes u_s+ \alpha(2\rho_n+\rho_s)I_N\right)\\
&= \frac{h^2}{2}\nabla_x.\left( \sqrt{\rho_s}\nabla_x^2 \sqrt{\rho_s}-\nabla_x\sqrt{\rho_s}\otimes \nabla_x\sqrt{\rho_s}\right)
\end{split}
\eee
where $I_N$ is the $N\times N$ identity matrix.

The momentum of each fluid is not preserved even at thermodynamic equilibrium; however, the total momentum obeys a conservation law. This means that, even at the equilibrium, the two parts of the fluid exchange some momentum.
\subsubsection{Energy}

We now integrate the kinetic equation against $|v|^2$:

\[
\begin{split}
\partial_t &\left( \int_{\R^N} |v|^2  \vphi \ud v \right) + \nabla_x. \left( \int_{\R^N} v|v|^2 \vphi \ud v \right) = -2 \rho_n u_n .\nabla_x V_1\\
&-\Big(\rho_s|u_s|^2 +\rho_s (2 V_2- V_1)- h^2\sqrt{\rho_s}\Delta_x \sqrt{\rho_s}-\rho_sV_1\Big)Q_2(\vphi).
\end{split}
\]
For the condensed part, the energy writes $h^2|\nabla_x\Psi|^2$, and satisfies
\be
\label{energypsi}
\begin{split}
 \partial_t h^2|\nabla_x\Psi|^2 &=h^3\Re(i\nabla_x(\Delta_x \Psi).\nabla_x\bar\Psi) -2h\Im(\bar\Psi\nabla_x\Psi) \nabla_x V_2\\
&\quad+\frac{h^2}{2}\nabla_x\rho_s.\nabla_x Q_2(\vphi)+h^2|\nabla_x\Psi|^2 Q_2(\vphi).
\end{split}
\ee
Computing
\bee
 h^3\Re(i\nabla_x(\Delta_x \Psi).\nabla_x\bar\Psi)
\eee
and since
\[
 h^2|\nabla_x\Psi|^2=\rho_s|u_s|^2+h^2|\nabla_x\sqrt{\rho_s}|^2,
\]
it follows that
\bee
\begin{split}
 \partial_t \left(\rho_s|u_s|^2+h^2|\nabla_x\sqrt{\rho_s}|^2\right)&+\nabla_x.(\rho_s|u_s|^2u_s)=-2\rho_s u_s \nabla_x V_2+\rho_s|u_s|^2Q_2(\vphi)\\
&+\frac{h^2}{2}\nabla_x\rho_s.\nabla_x Q_2(\vphi)+h^2|\nabla_x\sqrt{\rho_s}|^2Q_2(\vphi)\\
&+3h^2\nabla_x.(\nabla_x\sqrt{\rho_s}\otimes\nabla_x\sqrt{\rho_s}:u_s)\\
&-h^2\nabla_x.(\sqrt{\rho_s}u_s\Delta_x \sqrt{\rho_s} )+h^2\nabla_x.(\nabla_x\sqrt{\rho_s}\nabla_x.(\sqrt{\rho_s}u_s)).
\end{split}
\eee
On the other hand, the equation for the potential energy writes:
\bee
\begin{split}
\partial_t \frac{\alpha}{2}(2\rho_n+\rho_s)^2+&\nabla_x.(\alpha(2\rho_n+\rho_s)(2\rho_nu_n+\rho_su_s))\\
&=\alpha(2\rho_nu_n+\rho_su_s).\nabla_x(2\rho_n+\rho_s)-\alpha\rho_s(2\rho_n+\rho_s)Q_2(\vphi)\\
&=2\rho_n u_n\nabla_xV_1+2\rho_su_s\nabla_xV_2-\rho_sV_1Q_2(\vphi)
\end{split}
\eee
thanks to expressions (\ref{potentials}) of the potentials. If we add this equation the the equation for the energy of the condensed part (\ref{energypsi}), we note that we obtain a conservation equation, plus some terms coming from the interaction with the normal fluid. This conservation property is reminiscent from the energy conservation of the solutions to the Schr\"odinger equation.

Collecting all the pieces together, we get the full energy equation:
\begin{equation}
 \label{energy}
\begin{split}
 \partial_t &\left( \int_{R^N} |v|^2  \vphi \ud v +\rho_s|u_s|^2+h^2|\nabla_x\sqrt{\rho_s}|^2 +\frac{\alpha}{2}(2\rho_n+\rho_s)^2\right)\\
&+\nabla_x. \left( \int v|v|^2 \vphi \ud v +\rho_s|u_s|^2u_s+\alpha(2\rho_n+\rho_s)(2\rho_nu_n+\rho_su_s)\right)\\
&=h^2T(\vphi,\rho_s,u_s)
\end{split}
\end{equation}
with
\begin{equation*}
\begin{split}
T(\vphi,\rho_s,u_s)=&\nabla_x.(\sqrt{\rho_s}Q_2(\vphi)\nabla_x \sqrt{\rho_s})+3\nabla_x.(\nabla_x\sqrt{\rho_s}\otimes\nabla_x\sqrt{\rho_s}:u_s)\\
&- \nabla_x.(\sqrt{\rho_s}u_s\Delta_x\sqrt{\rho_s})+\nabla_x.(\nabla_x\sqrt{\rho_s}\nabla_x.(\sqrt{\rho_s}u_s)).
\end{split}
\end{equation*}

We check that the total energy is conserved, and that all the terms which do not have a clear physical interpretation formally vanish in the limit $h\to 0$.

\subsection{Thermodynamic equilibrium}

Letting $\ep\to 0$ in (\ref{modeleps}) suggests that, in the fast relaxation limit, $\vphi$ satisfies, for almost every $v\in \R^N$,
\[
 Q(\vphi)+\rho_sQ_1(\vphi)=0.
\]
We therefore need to characterize such functions:

\begin{prop}
 Let $\vphi\equiv \vphi(v)$ be a smooth, positive valued function on $\R^N$. Let $u_s \in \R^N$, $V_1\in\R$ be some given constants and let
\[
 \rho_n=\int_{v\in\R^N}\vphi(v)\ud v,\qquad u_n=\frac{1}{\rho_n}\int_{\R^N}v\vphi(v)\ud v.
\]
Define $Q(\vphi)$ and $Q_1(\vphi)$ by (\ref{Q}) and (\ref{Q1}). Assume that $h=0$, so that $C_2$ is the set of all the velocities $v_*,v_*'$ such that
\[
 \begin{cases}
  &u_s+v_*'=v+v_*\\
  &|u_s|^2+V_1+|v_*'|^2=|v|^2+|v_*|^2.
 \end{cases}
\]
Then,
\bee
 Q(\vphi)+\rho_sQ_1(\vphi)=0 
\eee
if and only if $\vphi$ is a bosonian, i.e.
\be
\label{equilibre}
\vphi(v)=\frac{M}{1-M},
\ee
where
\[
M=e^{-\frac{|v-u_n|^2-|u_s-u_n|^2+V1}{2T}},
\]
for some positive constant $T$.
\end{prop}

\begin{proof}
The identity
\[
 Q(\vphi)+\rho_sQ_1(\vphi)=0
\]
implies that
\[
 \int_{\R^N}\left(Q(\vphi)+\rho_sQ_1(\vphi)\right)\log \left(\frac{\vphi}{1+\vphi}\right)\ud v=0.
\]
Using the symmetries of the collision integrals, we have
\[
\begin{split}
 &\int_{\R^N}Q(\vphi)\log \left(\frac{\vphi}{1+\vphi}\right)\ud v\\
 &=\int_{\R^N}\int_{\R^N}\int_{S^{N-1}}(1+\vphi)(1+\vphi_*)(1+\vphi')(1+\vphi_*')\left(\frac{\vphi'}{1+\vphi'}\frac{\vphi_*'}{1+\vphi_*'}-\frac{\vphi}{1+\vphi}\frac{\vphi_*}{1+\vphi_*}\right)\\
&\times\left(\log\left(\frac{\vphi}{1+\vphi}\frac{\vphi_*}{1+\vphi_*}\right)-\log\left( \frac{\vphi'}{1+\vphi'}\frac{\vphi_*'}{1+\vphi_*'}\right)\right)\ud v\ud v_* \ud n
\end{split}
\]
and
\[
 \begin{split}
 & \int_{\R^N}Q_1(\vphi)\log \left(\frac{\vphi}{1+\vphi}\right)\ud v\\
 &=\int_{\R^N}\int_{C_2}(1+\vphi)(1+\vphi_*)(1+\vphi_*')\left(\frac{\vphi_*'}{1+\vphi_*'}-\frac{\vphi}{1+\vphi}\frac{\vphi_*}{1+\vphi_*}\right)\\
&\times\left(\log\left(\frac{\vphi}{1+\vphi}\frac{\vphi_*}{1+\vphi_*}\right)-\log\left(\frac{\vphi_*'}{1+\vphi_*'}\right)\right)\ud v\ud v_* \ud v_*'.
 \end{split}
\]
But for all $x,y>0$ we have
\[
 (x-y)(\log y-\log x) \leq 0
\]
with equality if and only if $x=y$. Hence both integrands are nonpositive functions, and we claim that:
\[
  Q(\vphi)+\rho_sQ_1(\vphi)=0  
\]
if and only if
\bee
Q(\vphi)=0 \quad\textrm{a.e. and}\quad Q_1(\vphi)=0 \quad\textrm{a.e.}.
\eee
For the same reason, $Q(\vphi)$ vanishes if and only if
\be
\label{bosonian}
\frac{\vphi'}{1+\vphi'}\frac{\vphi_*'}{1+\vphi_*'}=\frac{\vphi}{1+\vphi}\frac{\vphi_*}{1+\vphi_*}
\ee
for almost every $v,v_*\in\R^N$, $n\in S^{N-1}$, with $v',v_*'$ satisfying (\ref{microscons}), and $Q_1(\vphi)$ vanishes only when
\[
\frac{\vphi_*'}{1+\vphi_*'}=\frac{\vphi}{1+\vphi}\frac{\vphi_*}{1+\vphi_*}
\]
for almost every $v\in\R^N$ and $v_*,v_*'\in C_2$ with $h=0$ ($u_s \in \R^N$ and $V_1$ are fixed). 

We therefore have to solve the functional equation (\ref{bosonian}). Thanks to the classical Boltzmann theory, we know that if for all $v,v_*,  v', v_*'$ satisfying (\ref{microscons}) we have 
\[
 ff_*=f'f_*',
\]
then $f$ is a Maxwellian distribution, namely
\[
 f\equiv M=e^{(a|v|^2+b.v+c)}
\]
with $a,c\in\R$ and $b\in\R^N$. Taking $f=\frac{\vphi}{1+\vphi}$ gives that $Q(\vphi)$ vanishes if and only if
\[
\vphi=\frac{M}{1-M}.
\]
The Maxwellian function $M$ can be rewritten
\[
 M=e^{-\frac{|v-u_n|^2}{2T}+\delta}
\]
for some $T>0$ and $\delta\in\R$ using the macroscopic velocity, since
\[
 \int_{\R^N}v\vphi\ud v = \rho_n u_n.
\]

We can compute $\delta$ by plugging this into $Q_1(\vphi)=0$, which leads to
\[
 M_*'=MM_*
\]
for all $v\in\R^N,v_*,v_*'\in C_2$. We thus obtain:
\[
 \frac{|v_*'-u_n|^2}{2T}-\delta=\frac{|v-u_n|^2}{2T}-\delta+\frac{|v_*-u_n|^2}{2T}-\delta
\]
which gives, thanks to the definition of $C_2$,
\[
 \delta=\frac{|u_s-u_n|^2-V_1}{2T},
\]
which is the expected result.
\end{proof}

\begin{rem}
$M$ satisfies
\[
 \frac{\vphi}{1+\vphi}=M,
\]
so that
\[
 0<M<1.
\]
In particular,
\[
 |u_s-u_n|^2\leq\alpha(2\rho_n+\rho_s).
\]
\end{rem}

\begin{rem}
 In \cite{EMV03}, it is shown that the equilibrium solutions of the Boltzmann-Nordheim equation, which are the minimizers of the entropy, are the bosonian distributions (\ref{equilibre}) to which a Dirac function must be added if the mass of the initial data is greater than some critical value. In fact, for initial data having a big mass, the formation of a singularity taking the form of a Dirac function means that a condensate has been created in the fluid. Our goal here is to study a model (\ref{modele}) where a condensate has already occured, and then the corresponding mass has been deleted from the normal component of the fluid. Hence, the initial mass of the normal part of the fluid is assumed to be less than the critical value, and consequently the equilibrium solution is made only of the regular part.
\end{rem}

\section{Formal derivation of the fluid model}
\label{hydrodynamic limit}

Our goal in this section is to derive a fluid description of the Bose gas governed by (\ref{modele}). This will be achieved at formal level in three steps: we first perform the semiclassical limit $h\to 0$. It allows to pass from a quantum to a classical description of the condensed part. Then, rescaling the time and space as in (\ref{modeleps}), we perform the compressible Euler limit by the moment method. We get a set of equations on the masses and momentum of the two fluids, and an equation for the total energy. This system is closed. However, it is a system of $2N+3$ equations and as much unknowns. To simplify the subsequent study, we replace it in the third step with a smaller system using the isentropic approximation.

\subsection{Semiclassical limit}

When performed on the usual Schr\"odinger equation, the semiclassical limit $h\to 0$ allows to recover the laws of classical mechanics from quantum dynamics \cite{carlesbook}. Indeed, $h$ can be seen as a measure of how far we are from classical mechanics. In our case, the limit will allow us to treat the condensed part as a classical fluid.

We assume that all the quantities converge in suitable functional spaces as $h\to 0$. Moreover, we still write $\rho_s, u_s, \theta,\vphi, \rho_n, u_n$ the limits of these quantities. Under these assumptions, the equation for the phase (\ref{phase}) becomes
\bee
  \partial_t \theta = -\frac{|u_s|^2}{2}-V_2.
\eee
Hence, as $h\to 0$, $Q_2(\vphi)$ is changed into
\begin{equation*}
 Q_2(\vphi) = \int_{ C_3} \big[ \vphi(v')\vphi(v_*')-\vphi(v_*)(\vphi(v')+\vphi(v_*')+1)\big]\ud v_* \ud v' \ud v_*'
\end{equation*}
where the new $C_3$ is the set of all velocities $v_*, v', v_*'$ such that
\begin{equation*}
 \begin{cases}
  & u_s + v_* = v' + v_*' \\
  & |u_s|^2 -V_1+ |v_*|^2 = |v'|^2+|v_*'|^2.
 \end{cases}
\end{equation*}
Substituting $Q_2(\vphi)$ with its new expression, the mass equations (\ref{massus}), (\ref{massun}) and (\ref{masstotal}) write the same as before. The momentum equation for the normal fluid (\ref{momentumun}) is also unchanged. The equation for the momentum of the condensed part becomes:
\[
 \partial_t(\rho_s u_s) + \nabla_x.(\rho_s u_s \otimes u_s)= -\rho_s \nabla_x V_2+ \rho_s u_s Q_2(\vphi).
\]
Thus, the total momentum satisfies the conservation law
\bee
 \partial_t (\rho_nu_n+\rho_su_s) +\nabla_x.\left(\int_{\R^N} v\otimes v \vphi \ud v +\rho_s u_s \otimes u_s+ \alpha(2\rho_n+\rho_s)I_N\right)=0.
\eee
Letting $h\to 0$ in the energy equation (\ref{energy}) leads to:
\begin{equation*}
\begin{split}
 \partial_t &\left( \int_{R^N} |v|^2  \vphi \ud v +\rho_s|u_s|^2+\frac{\alpha}{2}(2\rho_n+\rho_s)^2\right)\\
&+\nabla_x. \left( \int v|v|^2 \vphi \ud v +\rho_s|u_s|^2u_s+\alpha(2\rho_n+\rho_s)(2\rho_nu_n+\rho_su_s)\right)=0.
\end{split}
\end{equation*}
Note that in the semiclassical limit, the total energy is conserved.

\subsection{Hydrodynamic limit}

We now perform the change of scale $(t,x)\mapsto (\frac{t}{\ep}, \frac{x}{\ep})$ where $\ep\to 0$ is a small parameter. As can be seen on the system (\ref{modeleps}), this scaling leaves the equations associated to the condensed part unchanged, excepted the collision term $Q_2(\vphi)$ which gets a factor $\frac{1}{\ep}$. The small parameter $\ep$ is known as the Knudsen number, which is proportional to the mean free path, that is, the average length covered by the particules of the gas between two collisions. Letting this parameter go to 0 means that the collisions occur on a time scale which is very small compared with the observation time scale, so that one can consider that the local thermodynamic equilibrium is reached almost instantaneously.

The principle of the moment method for the compressible Euler limit is as follows: multiplying the kinetic equation in (\ref{modeleps}) by $\ep$ and letting $\ep\to 0$ suggests that $\vphi$ converges to the bosonian (\ref{equilibre}). We thus replace $\vphi$ by its limiting value in the moment equations, namely, in the equations for the mass, momentum and kinetic energy of the normal fluid. We do the same for the equations which govern the mass, momentum and energy of the condensate, and we obtain a closed set of hydrodynamic equations describing the macroscopic evolution of a Bose gas.

We assume that all the macroscopic quantities converge as $\ep \to 0$ in suitable functional spaces. The only technical point is to compute some integrals, in particular those associated with the momentum and heat flux. We have the identity
\[
 \int_{\R^N} v\otimes v \frac{M}{1-M} \ud v= \rho_nu_n\otimes u_n + pI_N
\]
defining the pressure
\[
 p=\frac{1}{N}\int_{\R^N}|v|^2\frac{\beta e^{-\frac{|v|^2}{2T}}}{1-\beta e^{-\frac{|v|^2}{2T}}}
\]
where
\[
 \beta = e^{\frac{|u_s-u_n|^2-V_1}{2T}}\qquad\textrm{and}\qquad0<\beta<1.
\]
Using obvious symmetry properties, we deduce that
\[
 \int_{\R^N}|v|^2\frac{M}{1-M} \ud v=\rho_n |u_n|^2 + Np.
\]
In the same way, we have
\[
 \int_{\R^N}v|v|^2\frac{M}{1-M} \ud v=\rho_n |u_n|^2 u_n + (N+2)p u_n.
\]
At the end, the hydrodynamic limit of the system (\ref{modele}) is given by the following set of equations:
\be
\label{systemelimite}
\begin{cases}
&\partial_t \rho_n + \nabla_x. (\rho_n u_n)=0\\
&\partial_t \rho_s + \nabla_x.(\rho_s u_s)= 0\\
&\partial_t(\rho_n u_n)+\nabla_x.(\rho_n u_n\otimes u_n +pI_N)=-\alpha\rho_n\nabla_x(2\rho_n+\rho_s)\\
&\partial_t(\rho_s u_s)+\nabla_x.(\rho_s u_s\otimes u_s)=-\frac{\alpha}{2}\rho_s\nabla_x(2\rho_n+\rho_s)\\
&\partial_t(\frac{1}{2}\rho_n |u_n|^2 + \frac{1}{2}\rho_s |u_s|^2+ \frac{N}{2}p+\frac{\alpha}{4}(2\rho_n+\rho_s)^2))\\
&\quad+\nabla_x.\Big(\frac{1}{2}\rho_n |u_n|^2 u_n + \frac{1}{2}\rho_s |u_s|^2 u_s + \frac{(N+2)}{2}pu_n \\
&\qquad\qquad\qquad\qquad+\frac{\alpha}{2} (2\rho_n + \rho_s)(2\rho_n u_n + \rho_s u_s)\Big)=0
\end{cases}
\ee

This is a kind of two-phases Euler system, the second fluid (the superfluid) being pressureless. They do not exchange mass, while they exchange some momentum, contrary to what occurs in the model proposed by Landau \cite{landau}.

\subsection{The isentropic approximation}

The system (\ref{systemelimite}) is somehow very complex insofar as it is not a system of conservation laws, and defining weak solutions for such a system is complicated. Consequently, we approximate it with a simpler system, which can be written as a system of conservation laws. To do this, we first claim that the kinetic equation in (\ref{modeleps}) is endowed with a natural entropy $S(\vphi)$, which is defined by:
\[
 S(\vphi) = \frac{1}{\rho_n}\int_{\R^N}\left((1+\vphi)\log (1+\vphi)-\vphi\log \vphi \right)\ud v.
\]
We will show that in the case when that entropy is approximatively constant, we can compute the pressure law, thus reducing the number of unknowns, and obtain a closed system of conservation laws.

The entropy satisfies the following equation:
\be
\label{Htheorem}
\begin{split}
 \partial_t (\rho_n S(\vphi)) &+ \nabla_x.\left(\int_{\R^N}v\left((1+\vphi)\log (1+\vphi)-\vphi\log \vphi\right)\ud v\right) \\
&= -\frac{1}{\ep}\int_{\R^N}\left(Q(\vphi)+\rho_s Q_1(\vphi)\right)\log \left(\frac{\vphi}{1+\vphi}\right)\ud v \geq 0
\end{split}
\ee
which is analogous to Boltzmann's H-theorem. When $\ep \to 0$, the density function $\vphi$ goes to its equilibrium value $\frac{M}{1-M}$ and the entropy (now called $S$) becomes
\be
\label{S}
 \rho_n S =-\int_{\R^N}\left(\frac{M}{1-M}\log M +\log(1-M) \right).
\ee
Then, it comes
\[
 \int_{\R^N}v\left(\frac{1}{1-M}\log \frac{1}{1-M} -\frac{M}{1-M}\log\frac{M}{1-M}\right)\ud v=\rho_nSu_n
\]
and $S$ satisfies
\bee
 \partial_t (\rho_n S) + \nabla_x . (\rho_n S u_n)=0.
\eee
Moreover, thanks to the mass conservation for the normal fluid (first equation in the system (\ref{systemelimite})), we get
\be
\label{transportS}
\partial_t S + u_n.\nabla_x S =0
\ee
which is a transport equation. Note that $S\equiv S_0$, where $S_0$ is a constant, is a trivial solution of (\ref{transportS}).

\begin{rem}
\label{rementropy}
 In the fast relaxation limit $\ep\to 0$, the entropy dissipation is expected to concentrate on shock solutions, as occurs for the classical Boltzmann equation. Hence, inequality (\ref{Htheorem}) is expected to become an equality in the limit as long as smooth solutions are considered. In other words, (\ref{transportS}) should be a combination of the other equations of the system.
\end{rem}

Using (\ref{S}), we can compute the entropy with respect to  $\rho_n, p, \beta$ and $T$:
\[
  S = N\frac{p}{2T\rho_n}-\log \beta -\frac{T^{N/2}}{\rho_n}\int_{\R^N}\log (1-\beta e^{-\frac{|v|^2}{2}})\ud v.
\]
From the expressions of the density and pressure for the normal fluid:
\[
 \rho_n=T^{N/2}\int_{\R^N}\frac{\beta e^{-\frac{|v|^2}{2}}}{1-\beta e^{-\frac{|v|^2}{2}}}\ud v,\quad p= \frac{1}{N} T^{N/2+1} \int_{\R^N}|v|^2\frac{\beta e^{-\frac{|v|^2}{2}}}{1-\beta e^{-\frac{|v|^2}{2}}}\ud v,
\]
remarking that, thanks to an integration by parts,
\[
  \int_{\R^N}|v|^2 \frac{\beta e^{-\frac{|v|^2}{2}}}{1-\beta e^{-\frac{|v|^2}{2}}}\ud v =-\frac{N}{\beta}\int_{\R^N}\log(1-\beta e^{-\frac{|v|^2}{2}})\ud v,
\]
we get
\[
 S= \left( \frac{N}{2}+1 \right) \frac{p}{T\rho_n}-\log \beta.
\]
If we write
\be
\label{rhop}
 \rho_n=T^{N/2}F_0(\beta),\qquad p=\frac{1}{N}T^{N/2+1}F_2(\beta),
\ee
we obtain at the end
\[
 S = \left(\frac{1}{2}+\frac{1}{N}\right)\frac{F_2(\beta)}{F_0(\beta)}-\log \beta.
\]
It comes out that the entropy depends only on the variable $\beta$. If we show that $\beta \mapsto S(\beta)$ is one-to-one, we will be able to say that if $S$ remains constant, so is $\beta$; hence, thanks to (\ref{rhop}), we will obtain the following pressure law:
\be
\label{pressurelaw}
 p=\tilde c_N\rho_n^{\frac{N+2}{N}}.
\ee
with
\[
 \tilde c_N=\frac{1}{N}\frac{F_2(\beta_0)}{(F_0(\beta_0))^{1+2/N}}>0
\]
for some constant $\beta_0 \in (0,1)$.
\begin{lem}
\label{lemS}
 The function $\beta \in(0,1)\mapsto S(\beta)$ is one-to-one.
\end{lem}

\begin{proof}
 
We compute the derivative of $S$ with respect to $\beta$:
\[
 S'(\beta) = \left(\frac{1}{2}+\frac{1}{N}\right) \frac{F_2'(\beta)F_0(\beta)-F_2(\beta)F_0'(\beta)}{F_0(\beta)^2}-\frac{1}{\beta}.
\]
Hence, 
\bee
\begin{split}
 S'(\beta)=&\left(\frac{1}{2}+\frac{1}{N}\right)\frac{1}{F_0(\beta)^2}\Bigg( \int_{\R^N}|v|^2\frac{e^{-\frac{|v|^2}{2}}}{\left(1-\beta e^{-\frac{|v|^2}{2}}\right)^2}\ud v\int_{\R^N} \frac{ \beta e^{-\frac{|v|^2}{2}}}{1-\beta e^{-\frac{|v|^2}{2}}}\ud v\\
&-\int_{\R^N}|v|^2 \frac{ e^{-\frac{|v|^2}{2}}}{1-\beta e^{-\frac{|v|^2}{2}}}\ud v \int_{\R^N} \frac{\beta e^{-\frac{|v|^2}{2}}}{\left(1-\beta e^{-\frac{|v|^2}{2}}\right)^2}\ud v\Bigg)-\frac{1}{\beta}.\\
\end{split}
\eee
We want to show that $S'(\beta)$ is negative. To do this, we expand the integrands into power series. We get for the simplest one
\[
 \int_{\R^N} \frac{ e^{-\frac{|v|^2}{2}}}{1-\beta e^{-\frac{|v|^2}{2}}}\ud v = \int_{\R^N}\sum_{n=0}^\infty \beta^ne^{-(n+1)\frac{|v|^2}{2}}\ud v = \sum_{n=0}^\infty \beta^n \int_{\R^N}e^{-(n+1)\frac{|v|^2}{2}}\ud v,
\]
that is,
\[
  \int_{\R^N} \frac{ e^{-\frac{|v|^2}{2}}}{1-\beta e^{-\frac{|v|^2}{2}}}\ud v =\left(\int_{\R^N}e^{-\frac{|v|^2}{2}}\ud v\right) \sum_{n=0}^\infty \beta^n\frac{1}{(n+1)^{N/2}}.
\]
We compute the other integrals using the same tool, and the entropy derivative becomes
\[
 \begin{split}
  S'(\beta)&= \beta\left(\frac{1}{2}+\frac{1}{N}\right)\left(\int_{\R^N}|v|^2 e^{-\frac{|v|^2}{2}}\ud v\right)\left(\int_{\R^N}e^{-\frac{|v|^2}{2}}\ud v\right) \frac{1}{F_0(\beta)^2}\times\\
&\left(\left(\sum_{n=0}^\infty \beta^n\frac{1}{(n+1)^{N/2}}\right)^2-\sum_{n=0}^\infty \beta^n\frac{1}{(n+1)^{N/2+1}}\sum_{n=0}^\infty \beta^n\frac{1}{(n+1)^{N/2-1}} \right)-\frac{1}{\beta}.
 \end{split}
\]
We can calculate the products of the summations:
\[
 \left(\sum_{n=0}^\infty \beta^n\frac{1}{(n+1)^{N/2}}\right)^2=\sum_{n=0}^\infty c_n \beta^n
\]
with
\[
 c_n=\sum_{k=0}^n \frac{1}{(k+1)^{N/2}(n-k+1)^{N/2}},
\]
and
\[
 \sum_{n=0}^\infty \beta^n\frac{1}{(n+1)^{N/2+1}}\sum_{n=0}^\infty \beta^n\frac{1}{(n+1)^{N/2-1}}=\sum_{n=0}^\infty \tilde c_n \beta^n
\]
with
\[
 \tilde c_n=\sum_{k=0}^n \frac{1}{(k+1)^{N/2+1}(n-k+1)^{N/2-1}}.
\]
A simple calculation shows that
\[
 \tilde c_n \geq c_n.
\]
But since
\[
  S'(\beta)=\left(\frac{1}{2}+\frac{1}{N}\right)\int_{\R^N}|v|^2e^{-\frac{|v|^2}{2}}\ud v \int_{\R^N}e^{-\frac{|v|^2}{2}}\ud v\frac{\beta}{F_0(\beta)^2} \left( \sum_{n=0}^\infty (c_n-\tilde c_n)\beta^n \right) -\frac{1}{\beta},
\]
it comes
\[
 S'(\beta)<0,
\]
and thus $\beta \mapsto S(\beta)$ is one to one.
\end{proof}

From now on, we will assume that the dimension of the space is $N=3$. Moreover, we assume that the gas is translation invariant along the y- and z-axis, and that the y- and z-components of the velocities $u_n$ and $u_s$ are zero. Under these conditions, the system under consideration becomes one-dimensional
\be
\label{syst}
\begin{cases}
&\partial_t \rho_n + \partial_x (\rho_n u_n)=0\\
&\partial_t \rho_s + \partial_x(\rho_s u_s)= 0\\
&\partial_t(\rho_n u_n)+\partial_x(\rho_n u_n^2 +\tilde c \rho_n^{5/3})=-\alpha\rho_n\partial_x(2\rho_n+\rho_s)\\
&\partial_t(\rho_s u_s)+\partial_x(\rho_s u_s^2)=-\frac{\alpha}{2}\rho_s\partial_x(2\rho_n+\rho_s)\\
\end{cases}
\ee
where $\tilde c=\tilde c_3$. 

If $(\rho_n,\rho_s,u_n,u_s)$ is a regular solution of (\ref{syst}), then $(\rho_n,\rho_s,u_n,u_s,p\equiv\tilde c \rho_n^{5/3})$ will be a regular solution of (\ref{systemelimite}). However, this is not true any more when dealing with weak solutions. The reason is that the entropy  does not remain constant along shock waves as mentioned in remark \ref{rementropy}. Hence, system (\ref{syst}) is actually an approximation of system (\ref{systemelimite}), and not just a particular case of it.

\section{Hyperbolicity and regular solutions}
\label{hyperbolicity}

In view of giving some mathematical justifications to the derivation of the hydrodynamic model (\ref{systemelimite}) and of its isentropic approximation (\ref{syst}), it is necessary to have a good understanding of its structure and a good mathematical theory of existence and stability. We are therefore interested in this section in finding smooth solutions to (\ref{syst}).

In Section \ref{weak solutions}, we will deal with weak solutions of this system. As it will be more convenient to work with a system of conservation laws, we introduce right now the following system of equations:
\begin{equation}
\label{conservation}
 \begin{cases}
  & \partial_t \rho_n + \partial_x (\rho_n u_n)=0 \\
  & \partial_t \rho_s + \partial_x (\rho_s u_s)=0 \\
  & \partial_t (\rho_n u_n + \rho_s u_s)+ \partial_x (\rho_n u_n^2 + \rho_s u_s^2 + \tilde c \rho_n^{5/3}+ \frac{\alpha}{4}(2\rho_n+\rho_s)^2 )=0 \\
  & \partial_t (\frac{1}{2}\rho_n u_n^2 + \frac{1}{2}\rho_s u_s^2+ \frac{3}{2}\tilde c \rho_n^{5/3}+\alpha(\rho_n+\frac{1}{2}\rho_s)^2)\\
  &\quad\quad + \partial_x(\frac{1}{2}\rho_n u_n^3 + \frac{1}{2}\rho_s u_s^3 + \frac{5}{2}\tilde c\rho_n^{5/3}u_n+\alpha (\rho_n + \frac{1}{2}\rho_s)(2\rho_n u_n + \rho_s u_s))=0 \\
 \end{cases}
\end{equation}
Actually, the inversibility of the function
\bee
\begin{split}
 F: \R_+^* \times \R_+^* \times \R \times \R &\to \R_+^* \times \R_+^* \times \R \times \R_+^* \\
(\rho_n,\rho_s,u_n,u_s)&\mapsto (\rho_n,\rho_s,\rho_n u_n+\rho_s u_s, \frac{1}{2}\rho_n u_n^2 + \frac{1}{2}\rho_s u_s^2 \\
&\quad\quad\quad+\frac{3}{2}\tilde c \rho_n^{\frac{5}{3}}+\alpha(\rho_n+\frac{1}{2}\rho_s)^2)
\end{split}
\eee
(whose jacobian is $j_F = \rho_n \rho_s(u_s-u_n)$) on the sets
\[
 \mathcal O_+ = \left\{ U=\begin{pmatrix}\rho_n \\ \rho_s \\ u_n \\ u_s \end{pmatrix} \quad \textrm{s.t.}\quad \rho_n,\rho_s>0\quad\textrm{and}\quad u_n> u_s \right\}
\]
and
\[
 \mathcal O_- = \left\{ U=\begin{pmatrix}\rho_n \\ \rho_s \\ u_n \\ u_s \end{pmatrix} \quad \textrm{s.t.}\quad \rho_n,\rho_s>0\quad\textrm{and}\quad u_n< u_s \right\}
\]
is enough to say that on these sets, the smooth solutions to (\ref{syst}) are also solutions of (\ref{conservation}), and conversely. In addition, they are also solutions of
\be
\label{pb1}
\partial_t U + A \partial_x U=0
\ee
where
\begin{equation}
\label{matrice}
 A = \begin{pmatrix}
      u_n & 0 & \rho_n & 0 \\
      0 & u_s & 0 & \rho_s \\
      c\rho_n^{-\frac{1}{3}}+2\alpha & \alpha & u_n & 0 \\
      \alpha & \frac{\alpha}{2} & 0 & u_s
     \end{pmatrix}
\end{equation}
(with $c=\frac{5}{3}\tilde c>0$), and any of this problem is (strictly) hyperbolic provided the others are.

\subsection{Hyperbolicity}

For a system of the form (\ref{pb1}), the right property to study is hyperbolicity, that is, the possibility to diagonalize the matrix $A$ given by (\ref{matrice}). Indeed, this property allows to claim the existence and uniqueness of smooth solutions to (\ref{pb1}) (see \cite{dafermos}). It will also be useful when dealing with weak solutions.

Since system (\ref{syst}) is hyperbolic if and only if system (\ref{conservation}) or system (\ref{pb1}) is, it is enough to investigate the eigenvalues of the matrix $A$ given by (\ref{matrice}) to conclude on the hyperbolicity of any of these systems. The characteristic polynomial of $A$ writes
\bee
P_A(\lambda) = \left((\lambda-u_n)^2 -c\rho_n^\frac{2}{3}-2\alpha\rho_n\right)\left( (\lambda-u_s)^2 -\frac{\alpha}{2}\rho_s\right)-\alpha^2\rho_n \rho_s.
\eee
Even if it is of degree 4 and there are explicit formulae for the roots of polynomials of degree 4, they are very complicated and we are not able to use them here. Our analysis relies only on basic analytical tools to find cases where the polynomial has 4 different roots. The main result of this section is the following:

\begin{thm}
\label{thmhyper}
 Systems (\ref{syst}), (\ref{conservation}) and (\ref{pb1}) are strictly hyperbolic if $\rho_n, \rho_s >0$ and if one of these conditions holds:

 \be
\label{cond1}
 (u_n-u_s)^2 <c\rho_n^{2/3}
\ee
\be
 \label{cond2}
(u_n-u_s)^2 < \frac{c\frac{\alpha}{2}\rho_s\rho_n^{2/3}}{ c\rho_n^{2/3}+2\alpha\rho_n}
\ee
\be
\label{cond3}
\rho_n \leq \left(\frac{c}{2\alpha}\right)^3.
\ee
\end{thm}

Obviously, the condition (\ref{cond3}) is the most interesting, because it can be propagated by the solution, whereas the other two may fail to hold after a small time.

\begin{proof}
We work with $U=(\rho_n,\rho_s,u_n,u_s)$ and $\alpha$ fixed, with $\rho_n,\rho_s,\alpha>0$. We first notice that
\[
 \lim_{\lambda \to \pm \infty}P_A(\lambda)=+\infty.
\]
On the other hand, if we let
\[
\begin{cases}
 &\lambda_1^*=u_s-\sqrt{\frac{\alpha}{2}\rho_s}\\
 &\lambda_2^*=u_s+\sqrt{\frac{\alpha}{2}\rho_s}
\end{cases},
\]
we can remark that
\[
 \forall i=1,2 \qquad P(\lambda_i^*)=-\alpha^2 \rho_n\rho_s <0.
\]
But $u_s \in (\lambda_1^*,\lambda_2^*)$ satisfies
\[
 P_A(u_s)=-\frac{\alpha}{2}\rho_s\left((u_n-u_s)^2 -c\rho_n^{2/3}\right).
\]
Hence, if we assume that (\ref{cond1}) holds, we have
\[
P_A(u_s)>0.
\]
Since $P_A$ is a continous function of $\lambda$, we claim that it has four distinct roots $(\lambda_i)_{i=1..4}$ depending on $U$ and $\alpha$ and satisfying (see Figure \ref{polyno})
\[
 \lambda_1<u_s-\sqrt{\frac{\alpha}{2}\rho_s}<\lambda_2<u_s<\lambda_3<u_s+\sqrt{\frac{\alpha}{2}\rho_s}<\lambda_4.
\]

\begin{figure}[htbp]
  \includegraphics{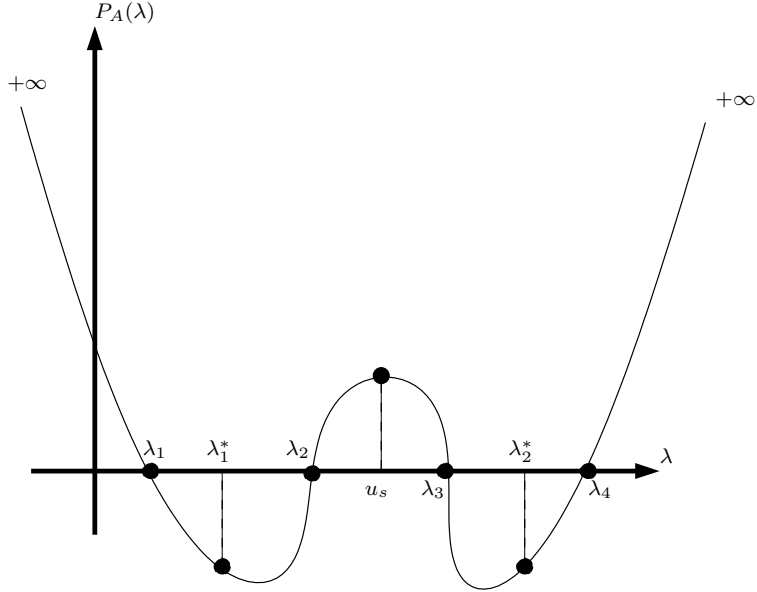}
\caption{In the case when $P_A(u_s)>0$, the polynom $P_A$ has four distinct roots.}
\label{polyno}
\end{figure}

In the same way, we let
\[
\begin{cases}
 &\lambda_3^*=u_n-\sqrt{c\rho_n^{2/3}+2\alpha\rho_n}\\
 &\lambda_4^*=u_n+\sqrt{c\rho_n^{2/3}+2\alpha\rho_n}
\end{cases};
\]
then
\[
 P_A(\lambda_3^*)=P_A(\lambda_4^*)=-\alpha^2\rho_n\rho_s<0
\]
and
\[
 P_A(u_n)=c\frac{\alpha}{2}\rho_s\rho_n^{2/3}-(u_n-u_s)^2(c\rho_n^{2/3}+2\alpha\rho_n).
\]
Thus, under the condition (\ref{cond2}), $P_A$ has four distinct roots $(\lambda_i)_{i=1..4}$ wich satisfy
\[
 \lambda_1<u_n-\sqrt{c\rho_n^{2/3}+2\alpha\rho_n}<\lambda_2<u_n<\lambda_3<u_n+\sqrt{c\rho_n^{2/3}+2\alpha\rho_n}<\lambda_4.
\]
We have then proved that under the condition (\ref{cond1}) or (\ref{cond2}), the systems (\ref{syst}), (\ref{conservation}) and (\ref{pb1}) are strictly hyperbolic.

We now assume that neither (\ref{cond1}) nor (\ref{cond2}) are verified, that is,
\be
\label{nocond1}
 (u_n-u_s)^2 \geq c\rho_n^{2/3}
\ee
and
\be
\label{nocond2}
(u_n-u_s)^2 \geq \frac{c\frac{\alpha}{2}\rho_s\rho_n^{2/3}}{c\rho_n^{2/3}+2\alpha\rho_n}.
\ee
Thus, $u_n \neq u_s$, $P_A(u_n)\leq 0$ and $P_A(u_s)\leq 0$. Let us study the sign of $P_A(\frac{u_n+u_s}{2})$:
\bee
P_A\left(\frac{u_n+u_s}{2}\right)=\frac{1}{16}(u_n-u_s)^4 -\frac{1}{4}\left(\frac{\alpha}{2}\rho_s +c\rho_n^{2/3}+2\alpha\rho_n\right)(u_n-u_s)^2 +c\frac{\alpha}{2}\rho_s\rho_n^{2/3}.
\eee
Let $Q(X)$ be the polynomial defined by
\[
 Q(X)=\frac{1}{16}X^2 -\frac{1}{4}\left(\frac{\alpha}{2}\rho_s +c\rho_n^{2/3}+2\alpha\rho_n\right)X +c\frac{\alpha}{2}\rho_s\rho_n^{2/3}.
\]
Its discriminant is
\[
 \delta =\frac{1}{16}\left(\frac{\alpha}{2}\rho_s-c\rho_n^{2/3}+2\alpha\rho_n \right)^2 +c\frac{\alpha}{2}\rho_n^{5/3}>0,
\]
and its two roots write
\[
 \beta_-=\frac{1}{32}\left(\frac{1}{4}\left(\frac{\alpha}{2}\rho_s +c\rho_n^{2/3}+2\alpha\rho_n\right)-\sqrt{\delta}\right)
\]
\[
 \beta_+=\frac{1}{32}\left(\frac{1}{4}\left(\frac{\alpha}{2}\rho_s +c\rho_n^{2/3}+2\alpha\rho_n\right)+\sqrt{\delta}\right).
\]
Hence,
\[
 P\left(\frac{u_n+u_s}{2}\right)>0 \ssi (u_n-u_s)^2 > \beta_+ \quad\textrm{or}\quad(u_n-u_s)^2 < \beta_-.
\]
Let us assume that (\ref{cond3}) holds, or equivalently, that
\[
 2\alpha\rho_n\leq c\rho_n^{2/3}.
\]
Then,
\[
 \beta_+ < \frac{1}{64}\left(\frac{\alpha}{2}\rho_s +3c\rho_n^{2/3}\right).
\]
But, thanks to (\ref{nocond1}) and (\ref{nocond2}), we know that
\[
 (u_n-u_s)^2>c\rho_n^{2/3}
\]
and
\[
  (u_n-u_s)^2>\frac{\alpha}{4}\rho_s,
\]
from which we deduce that
\[
 (u_n-u_s)^2>\beta_+.
\]
Moreover, we notice that $P_A'(u_s)P_A'(u_n)\neq 0$. We can conclude that the systems (\ref{syst}), (\ref{conservation}) and (\ref{pb1}) are strictly hyperbolic under the condition (\ref{cond1}), or (\ref{cond2}), or under the only condition (\ref{cond3}). In this case, the eigenvalues of A satisfy
\[
 \lambda_1< u_n \leq\lambda_2<\frac{u_n+u_s}{2}<\lambda_3\leq u_s<\lambda_4
\]
if $u_n<u_s$, or the same thing inverting $u_n$ and $u_s$ if $u_s<u_n$.
\end{proof}

\subsection{Entropy}

We are now interested in finding a convex entropy to system (\ref{syst}), that is, a convex function $E(U)$ such that there exists a function $G(U)$ which satisfies
\bee
 \partial_t E(U) + \partial_x G(U)=0.
\eee
Our research is motivated by the fact that a strictly hyperbolic system having a strictly convex entropy admits regular solutions (see \cite{dafermos} for example).

The entropy we are looking for is given by the physical energy:
\[
 E(U)=\frac{1}{2}\rho_n u_n^2 + \frac{1}{2}\rho_s u_s^2+ \frac{9}{10}c \rho_n^{\frac{5}{3}}+\alpha(\rho_n+\frac{1}{2}\rho_s)^2.
\]
We can find the equation satisfied by E just from (\ref{pb1}); it writes
\[
 \partial_t E(U) + \partial_x G(U) =0
\]
with
\[
 G(U) = \frac{1}{2}\rho_n u_n^3 + \frac{1}{2}\rho_s u_s^3 + \frac{3}{2}c\rho_n^\frac{5}{3}u_n  +\alpha (\rho_n + \frac{1}{2}\rho_s)(2\rho_n u_n + \rho_s u_s).
\]
Let us now study its convexity. Its Hessian matrix is given by
\bee
\nabla^2_{U,U}E(U)= \begin{pmatrix}
                 c\rho_n^{-\frac{1}{3}}+2\alpha & \alpha           & u_n    & 0   \\
		 \alpha                         & \frac{\alpha}{2} & 0      & u_s \\
                 u_n                            & 0                & \rho_n & 0   \\
                 0                              & u_s              & 0      & \rho_s
                \end{pmatrix}.
\eee
If we denote $X=(x_1,x_2,x_3,x_4)^T$, where $^T$ is the transposition operator, then
\[
 \begin{split}
  X^T (\nabla^2_{U,U}E(U)) X=&\rho_s \left( x_4 + \frac{u_s}{\rho_s}x_2\right)^2 + \rho_n \left( x_3 + \frac{u_n}{\rho_n}x_1\right)^2\\
&\qquad\qquad + \left(\frac{\alpha}{2}-\frac{u_s^2}{\rho_s} \right)\left(x_2 + \frac{2\alpha\rho_s}{\alpha\rho_s-2u_s^2}x_1 \right)^2 \\
&\qquad\qquad+\left(c\rho_n^{-\frac{1}{3}} +2\alpha-\frac{u_n^2}{\rho_n}-\frac{2\alpha^2\rho_s}{\alpha\rho_s-2u_s^2}\right)x_1^2.
 \end{split}
\]
Hence, $E$ will be strictly convex uniformly on every compact subset of $\mc O$ under the conditions 
\[
 \begin{cases}
  &\displaystyle{\frac{\alpha}{2}-\frac{u_s^2}{\rho_s}>0}\\
  &\displaystyle{c\rho_n^{-\frac{1}{3}} +2\alpha-\frac{u_n^2}{\rho_n}-\frac{2\alpha^2\rho_s}{\alpha\rho_s-2u_s^2}>0}
 \end{cases}
\]
which are equivalent to
\be
\label{entropy}
 \begin{cases}
  &\displaystyle{u_s^2 <\frac{\alpha}{2}\rho_s}\\
  &\displaystyle{ P_A(0) >0}
 \end{cases}.
\ee
We can now write the following theorem of existence of solutions:

\begin{thm}
 Let $\alpha>0$, and
\[
\begin{split}
\mathcal A = &\bigg\{(\rho_n,\rho_s,u_n,u_s)^T, \quad\rho_n,\rho_s>0, ~~(\ref{cond1}), (\ref{cond2})~\textrm{or}~ (\ref{cond3})~ \textrm{holds},\\
&\qquad u_s^2 <\frac{\alpha}{2}\rho_s, \quad P_A(0) >0,\quad u_n\neq u_s \bigg\}.
\end{split}
\]
Let $U_0 \in C^1(\R; \Omega)$ where $\Omega$ is a compact subset of $\mathcal A$, such that $\partial_x U_0 \in H^l$ for some $l>\frac{1}{2}$. 
Then there exists $T_\infty$, such that $0<T_\infty \leq \infty$, and a unique function $U\in C^1([0,T_\infty)\times \R;\mathcal A)$, which is a classical solution of the problem (\ref{pb1}) on $[0,T_\infty)$ with the initial condition
\[
 U(0,x)=U_0(x),\qquad\forall x\in \R,
\]
 and of the problem (\ref{conservation}) with same initial condition. Moreover,
\[
 \partial_x U(.,t) \in C^0([0,T_\infty);H^l).
\]
\end{thm}

\section{Weak solutions and the Riemann problem}
\label{weak solutions}

In Section \ref{hyperbolicity}, we have proved the existence of smooth solutions to the problem (\ref{syst}), or equivalently to (\ref{conservation}) or (\ref{pb1}). However, these solutions are not satisfactory. First, they exist only for small times, while we would like to have global solutions. Moreover, physical experiments often show non regular behaviours to physical problems. The solutions of hyperbolic systems having a physical meaning should therefore be authorized to be discontinous. It is then natural to define a weaker notion of solutions. 

In this section, we will recall the usual definition of weak solutions to a system of conservation laws, and state an existence and uniqueness result for the system (\ref{conservation}). We will then give some details on the qualitative behaviour of such a solution by studying the characteristic fields and the shock curves.

\subsection{Existence of weak solutions}

When dealing with weak solutions, systems (\ref{syst}), (\ref{conservation}) and (\ref{pb1}) are no more equivalent. Indeed, the usual definition of weak solutions for an hyperbolic system is valid only for systems of conservation laws. We consequently restrict our analysis to the system (\ref{conservation}), which we rewrite here:
\begin{equation}
\label{conservation2}
 \begin{cases}
  & \partial_t \rho_n + \partial_x (\rho_n u_n)=0 \\
  & \partial_t \rho_s + \partial_x (\rho_s u_s)=0 \\
  & \partial_t (\rho_n u_n + \rho_s u_s)+ \partial_x (\rho_n u_n^2 + \rho_s u_s^2 + \tilde c \rho_n^{5/3}+ \frac{\alpha}{4}(2\rho_n+\rho_s)^2 )=0\\
  & \partial_t (\frac{1}{2}\rho_n u_n^2 + \frac{1}{2}\rho_s u_s^2+ \frac{3}{2}\tilde c \rho_n^{5/3}+\alpha(\rho_n+\frac{1}{2}\rho_s)^2)\\
  &\quad\quad + \partial_x(\frac{1}{2}\rho_n u_n^3 + \frac{1}{2}\rho_s u_s^3 + \frac{5}{2}\tilde c\rho_n^{5/3}u_n +\alpha (\rho_n + \frac{1}{2}\rho_s)(2\rho_n u_n + \rho_s u_s))=0\\
 \end{cases}
\end{equation}
With obvious notations, we can write it
\bee
\partial_t F(U) +\partial_x H(U)=0.
\eee
We define weak solutions in the usual way:

\begin{defi}
 Let $U_0=(\rho_n^0,\rho_s^0,u_n^0,u_s^0)\in (L^\infty(\R))^4$. We say that the bounded vector field $U=(\rho_n,\rho_s,u_n,u_s)$ is a weak solution of system (\ref{conservation2}) supplemented with the initial condition $U_0$ if for all $\Phi \in (C^1(\R_+\times\R))^4$ compactly supported we have
\bee
\begin{split}
\int_0^{+\infty} \!\!\!\int_\R &\Big(F(U(t,x)).\partial_t \Phi(t,x) + H(U(t,x)).\partial_x \Phi(t,x)\Big)\ud x\ud t \\
&+\int_\R F(U_0(x)) \Phi (0,x)\ud x =0.
\end{split}
\eee

Moreover, we say that $U$ is a solution of the Riemann problem if the initial condition is of the form
\[
 U_0(x) = \begin{cases}
        U_0^-\quad \textrm{if}~x<0\\
	U_0^+\quad \textrm{if}~x>0\\
       \end{cases}
\]
where $U_0^\pm$ are two constant states in $\R^4$.
\end{defi}

We now state the existence and uniqueness result:

\begin{thm}
 Let $\bar U =(\bar \rho_n,\bar\rho_s,\bar u_n,\bar u_s)$ be a constant state such that $\bar\rho_n,\bar\rho_s>0$ and $\bar u_n \neq \bar u_s$. Let $U^0=(\rho_n^0,\rho_s^0,u_n^0,u_s^0)$ be a bounded function of $x$ having small total variation and such that $\|U^0-\bar U\|_{L^\infty}$ is sufficiently small; in particular, there exists a constant $r>0$ such that $\rho_n^0>r$ and $\rho_s^0>r$, and such that $u_n^0-u_s^0>r$ almost everywhere, or $u_s^0-u_n^0>r$ a.e.. Then the system (\ref{conservation2}) endowed with the intial condition $U^0$ admits a unique weak global entropic (in the sense of Liu) solution.
\end{thm}
Weak solutions of hyperbolic equations are not unique in general. The entropy criterion of Liu is a way of selecting physical solutions, and weak solutions are unique in that class. See \cite{baiti} for details and more general results on uniqueness of solutions to hyperbolic systems. General existence theorems for hyperbolic systems can be found in \cite{bressan} or \cite{glf}. Let us just mention that this theorem holds true because of the hyperbolicity condition (\ref{cond3}) wich is propagated by the solutions.

\subsection{Nature of the characteristic fields}

Let $\lambda$ be an eigenvalue of $A$, that is, a solution of
\be
\label{pol}
 \left((\lambda-u_n)^2 -c\rho_n^\frac{2}{3}-2\alpha\rho_n\right)\left( (\lambda-u_s)^2 -\frac{\alpha}{2}\rho_s\right)-\alpha^2\rho_n \rho_s=0.
\ee
An associated eigenvector is
\[
 X=\begin{pmatrix}
    (\lambda -u_s)^2 -\frac{\alpha}{2}\rho_s \\
    \alpha \rho_s \\
    \frac{1}{\rho_n}(\lambda -u_n)((\lambda-u_s)^2 -\frac{\alpha}{2}\rho_s)\\
    \alpha(\lambda-u_s)
   \end{pmatrix}.
\]
We can compute the derivatives of $\lambda$ directly from (\ref{pol}) and then we get, letting
\[
 K=-\frac{1}{2}P'_A(\lambda)\neq 0,
\]
\bee
\nabla_U \lambda=
\begin{pmatrix}
 \frac{1}{2K}(\frac{2}{3}c\rho_n^{-1/3}+2\alpha)(\frac{\alpha}{2}\rho_s -(\lambda-u_s)^2)\\
 \frac{1}{4K}(c\rho_n^{2/3}+2\alpha\rho_n-(\lambda-u_n)^2)\\
 \frac{1}{K}(\lambda-u_n)(\frac{\alpha}{2}\rho_s-(\lambda-u_s)^2)\\
 \frac{1}{K}(\lambda-u_s)(c\rho_n^{2/3}+2\alpha\rho_n-(\lambda-u_n)^2) 
\end{pmatrix}.
\eee
Thus, we get for the associated characteristic field
\bee
\begin{split}
 \nabla_U\lambda.X &=\frac{1}{K}\Bigg[ -\frac{1}{2}\left(\frac{2}{3}c\rho_n^{-1/3}+2\alpha\right)\left(\frac{\alpha}{2}\rho_s -(\lambda-u_s)^2\right)^2 \\
                   &\qquad\qquad+\frac{\alpha}{4}\rho_s\left(c\rho_n^\frac{2}{3}+2\alpha\rho_n-(\lambda-u_n)^2 \right)\\
                   &\qquad\qquad-\frac{1}{\rho_n}(\lambda-u_n)^2\left((\lambda-u_s)^2 -\frac{\alpha}{2}\rho_s\right)^2\\ &\qquad\qquad+\alpha(\lambda-u_s)^2\Big(c\rho_n^\frac{2}{3}+2\alpha\rho_n-(\lambda-u_n)^2 \Big)\Bigg].
\end{split}
\eee
It is easy to see that if
\be
\label{condition}
 (\lambda-u_n)^2- c\rho_n^\frac{2}{3}-2\alpha\rho_n>0,
\ee
holds, then the field is genuinely nonlinear. Since
\[
 P_A(u_n-\sqrt{c\rho_n^\frac{2}{3}+2\alpha\rho_n})=P_A(u_n+\sqrt{c\rho_n^\frac{2}{3}+2\alpha\rho_n})=-\alpha^2\rho_n\rho_s<0,
\]
condition (\ref{condition}) always holds for the fields 1 and 4 associated with the extremal eigenvalues. Hence these fields are genuinely nonlinear, and they will induce the formation of shock or rarefaction waves. Nevertheless, we have not been able to establish anything for the fields 2 and 3.

\subsection{Shock curves}

In this part, we are interested in the shock solutions of the Riemann problem for system (\ref{conservation}). The solutions we are looking for are of the form
\bee
U(t,x)=
\begin{cases}
 & U^- \quad \textrm{if}~x<\sigma t \\
 & U^+ \quad \textrm{if}~x>\sigma t 
\end{cases}
\eee
with $U^\pm = (\rho_n^\pm, \rho_s^\pm, u_n^\pm, u_s^\pm)$ the left and right constant states, and $\sigma \in \R$ the shock speed. They must satisfy the Rankine-Hugoniot equations, which write for our system
\be
\label{RH}
\begin{split}
&\sigma
\begin{bmatrix}
 \rho_n \\ \rho_s \\\rho_n u_n +\rho_s u_s \\ \frac{1}{2}\rho_n u_n^2 + \frac{1}{2}\rho_s u_s^2+ \frac{3}{2}\tilde c \rho_n^{5/3}+\alpha(\rho_n+\frac{1}{2}\rho_s)^2
\end{bmatrix}\\
&\qquad= \begin{bmatrix}
   \rho_n u_n \\ \rho_s u_s \\ \rho_n u_n^2 + \rho_s u_s^2 + \tilde c \rho_n^{5/3} + \frac{\alpha}{4}(2\rho_n+\rho_s)^2 \\ \frac{1}{2}\rho_n u_n^3 + \frac{1}{2}\rho_s u_s^3 + \frac{5}{2}\tilde c\rho_n^{5/3}u_n  +\alpha (\rho_n + \frac{1}{2}\rho_s)(2\rho_n u_n + \rho_s u_s)
  \end{bmatrix}.
\end{split}
\ee
Here, the bracket means
\[
 [X]=X^+ -X^-.
\]
Let
\[
 w_n=u_n-\sigma
\]
and
\[
 w_s = u_s -\sigma.
\]
Then, the first two equations of (\ref{RH}) write simply
\be
\label{RH1}
\begin{cases}
 &[\rho_n w_n]=0\\
 &[\rho_s w_s]=0.
\end{cases}
\ee
Hence, if we know the left state $U^-=(\rho_n^-,\rho_s^-,u_n^-,u_s^+)$, the shock speed $\sigma$ and the right densities $\rho_n^+, \rho_s^+$, it is easy to compute the right velocities $u_n^+$ and $u_s^+$.

The third equation in (\ref{RH}) writes
\[
 [\rho_nu_n^2 -\rho_nu_n\sigma +\rho_su_s^2 -\rho_su_s\sigma +\tilde c \rho_n^{5/3}+\frac{\alpha}{4}(2\rho_n+\rho_s)^2]=0.
\]
Using (\ref{RH1}) and letting
\[
 M_n = \rho_n^-w_n^- = \rho_n^+w_n^+,\qquad\qquad M_s=\rho_s^-w_s^-=\rho_s^+w_s^+,
\]
and
\[
 \tau_n=\frac{1}{\rho_n},\qquad\qquad\tau_s=\frac{1}{\rho_s},
\]
it becomes
\be
\label{RH3}
 M_n^2[\tau_n]+M_s^2[\tau_s]+\tilde c [\rho_n^{5/3}]+\frac{\alpha}{4}[(2\rho_n+\rho_s)^2]=0.
\ee
In the same way, the fourth Rankine-Hugoniot relation (\ref{RH}) becomes, using (\ref{RH3}):
\be
 \label{RH4}
M_n^3 [\tau_n^2]+M_s^3[\tau_s^2]+5 M_n[\rho_n^{2/3}]+\alpha(2M_n+M_s)[2\rho_n+\rho_s]=0.
\ee
Let
\bee
J(\rho_n,\rho_s,u_n,u_s,\sigma)=\begin{pmatrix}
                                 \rho_n(u_n-\sigma)\\
				 \rho_s(u_s-\sigma)\\
				 \rho_n(u_n-\sigma)^2 +\rho_s(u_s-\sigma)^2 + \tilde c \rho_n^{5/3} +\frac{\alpha}{4}(2\rho_n+\rho_s)^2\\
				 \rho_n(u_n-\sigma)^3 +\rho_s(u_s-\sigma)^3 + \frac{5}{2} \rho_n^{5/3}(u_n-\sigma)\\ \qquad\qquad+\frac{\alpha}{2}(2\rho_n+\rho_s)\left(2\rho_n(u_n-\sigma)+\rho_s(u_s-\sigma)\right)
                                \end{pmatrix}.
\eee
Equations (\ref{RH1}), (\ref{RH3}), (\ref{RH4}) imply that the Rankine-Hugoniot relations are satisfied if and only if
\[
 J(\rho_n^+,\rho_s^+,u_n^+,u_s^+,\sigma)= J(\rho_n^-,\rho_s^-,u_n^-,u_s^-,\sigma).
\]
Thanks to the implicit functions theorem, for any left state $U^-=(\rho_n^-,\rho_s^-,u_n^-,u_s^-)$ and shock velocity $\sigma_0$ such that $D_{(\rho_n,\rho_s,u_n,u_s)}J(\rho_n^-,\rho_s^-,u_n^-,u_s^-,\sigma_0)$ is invertible (take for example $\rho_n^-=\rho_s^-=u_n^-=1$, $u_s^-=\sigma_0=0$), there exist some right states $U^+(\sigma)$ close to $U^-$ for  $\sigma$ in a neighbourhood of $\sigma_0$ such that 
\[
 J(U^+(\sigma),\sigma)=J(U^-,\sigma),
\]
that is, system (\ref{conservation2}) admits some shock solutions, which can be parametrized by the velocity $\sigma$. Note that the entropy criterion will select just half of each shock curve.

Such shock solutions can be observed with numerical simulations. Examples of numerical results using a Lax-Friedrichs scheme (see \cite{godrav} for example) are given by Figure \ref{rhon} which represents the quantity $\rho_n$. The figure is made of three graphics which show the results of the numerical scheme using a spatial grid made of 1000, 5000 and 10000 points. The time step is taken so that
\[
 \frac{\Delta t}{\Delta x}=5.
\]
As expected, the solutions show 5 different states separated by 4 waves. This contrasts with the solutions of the Euler equations which cannot exhibit more than 3 waves in general, 2 waves in the isentropic case. The presence of the superfluid part induces therefore some changes even in the behaviour of the normal part of the fluid compared to a usual (non quantum) fluid.
\begin{figure}
 \begin{center}
 \includegraphics[width=12cm]{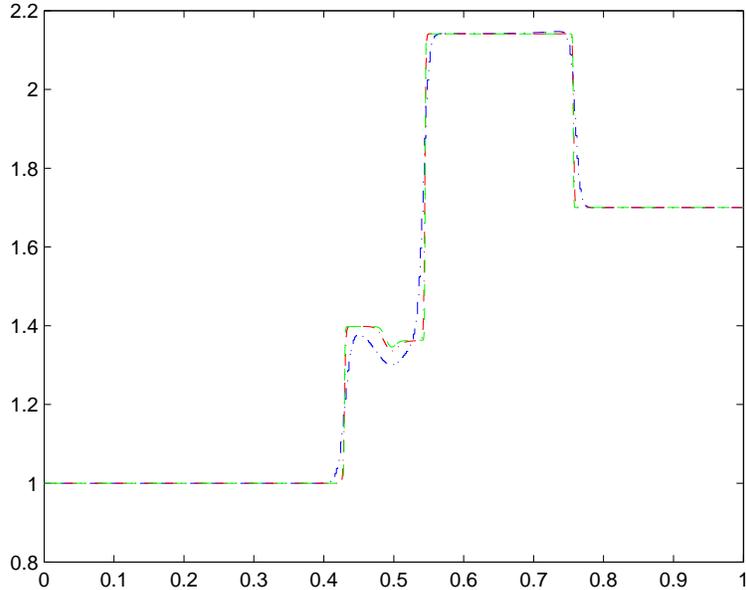}
\end{center}
\caption{The mass of the normal fluid $\rho_n$.}
\label{rhon}
\end{figure}

\section*{Conclusions}
We have derived a two-fluids model for a quantum gas from a system consisting of a Boltzmann-like and a Gross-Pitaevskii equations. Our derivation is formal, and it seems to be very difficult to obtain a rigourous proof of convergence, at least until there is a good existence framework for the quantum Boltzmann-Bose equation. However, we have shown that our hydrodynamic model is well-posed, which is the first step in a good understanding of the limiting process.

The limit system (\ref{conservation}) that we propose does not seem to exist in the litterature. It remains to see whether it provides a good description of the evolution of a superfluid, before any further mathematical exploration of its properties.

\section*{Acknowledgements} The author would like to thank S. Cordier for his help on numerical simulations.

\end{document}